\documentclass[11pt]{article}

\usepackage{fullpage}

\usepackage{dsfont}
\usepackage{latexsym}
\usepackage{natbib}
\usepackage{amsmath}
\usepackage{amsthm}
\usepackage{amssymb}
\usepackage{amsfonts}
\usepackage{mathrsfs}
\usepackage{aliascnt}
\usepackage{graphicx}
\usepackage{enumerate}

\usepackage[pdfpagelabels,pdfpagemode=None]{hyperref}
\usepackage{algorithmic}
\usepackage{algorithm}

\newtheorem{theorem}{Theorem}%

\newaliascnt{lemma}{theorem}
\newtheorem{lemma}[lemma]{Lemma}%
\aliascntresetthe{lemma}

\newaliascnt{claim}{theorem}
\newtheorem{claim}[claim]{Claim}%
\aliascntresetthe{claim}

\newaliascnt{corollary}{theorem}
\aliascntresetthe{corollary}

\newaliascnt{proposition}{theorem}
\aliascntresetthe{proposition}

\newaliascnt{algo}{procedure}
\aliascntresetthe{algo}

\theoremstyle{definition}
\newtheorem{definition}{Definition}

\newaliascnt{remark}{theorem}

\aliascntresetthe{remark}

\newcommand{\AutoAdjust}[3]{\mathchoice{ \left #1 #2  \right #3}{#1 #2 #3}{#1 #2 #3}{#1 #2 #3} }

\newcommand{\InBrackets}[1]{\AutoAdjust{[}{#1}{]}}% {\left[{#1}\right]}

\newcommand{\Prx}[2][]{\operatorname{\mathbf{Pr}}_{#1}\InBrackets{#2}}

\newcommand{\eps}{\epsilon}
\newcommand{\defeq}{\mathrel{\mathop:}=}

\newcommand{\property}{\mathcal P}
\newcommand{\domain}{D}
\newcommand{\range}{R}
\newcommand{\tfree}{\textsc{T-Free }}
\DeclareMathOperator{\dist}{dist}
\newcommand{\GF}{\mathbb F}
\newcommand{\Ftwo}{\GF_2}
\newcommand{\Ftwon}{\GF_2^{\ncoord}}
\newcommand{\perm}{\pi}
\DeclareMathOperator{\sym}{Sym}
\newcommand{\usp}{U}
\newcommand{\puzzle}{u}
\newcommand{\ncoord}{k}
\newcommand{\npuzzle}{m}

\newcommand{\ringsize}{M}
\newcommand{\puzsize}{N}
\newcommand{\puz}{I}
\newcommand{\puzset}{\mathcal I}

\title{Improved Lower Bounds for Testing Triangle-freeness in Boolean Functions via Fast Matrix Multiplication}

\author{Hu Fu\\
Microsoft Research \\
New England Lab \\
{\tt hufu@microsoft.com}
\and Robert Kleinberg \\
Cornell University \\
Dept.\@ of Computer Science \\
{\tt rdk@cs.cornell.edu}
}

\begin{document}

\begin{titlepage}
\maketitle

\begin{abstract}
Understanding the query complexity for testing linear-invariant properties has been a central open problem in the study of
algebraic property testing.  Triangle-freeness in Boolean functions is a simple property whose testing complexity is
unknown.  Three Boolean functions $f_1$, $f_2$ and $f_3: \Ftwon \to \{0, 1\}$ are said to be triangle free if there is
no $x, y \in \Ftwon$ such that $f_1(x) = f_2(y) = f_3(x + y) = 1$.  This property is known to be strongly testable
\citep{Gre05}, but the number of queries needed is upper-bounded only by a
tower of twos whose height is polynomial in $1 / \eps$, where $\eps$ is the distance between the tested function triple
and triangle-freeness, i.e.,  the minimum fraction of function values that need to be modified to make the triple
triangle free.  A lower bound of $(\tfrac 1 \eps)^{2.423}$ for any one-sided tester was given by \citet{BX10}.  In this work we improve this
bound to $(\tfrac 1 \eps)^{6.619}$.  Interestingly, we prove this by way of a combinatorial construction called
\emph{uniquely solvable puzzles} that was at the heart of \citet{CW90}'s renowned matrix multiplication algorithm.
\end{abstract}
\thispagestyle{empty}
\end{titlepage}

\newpage

\section{Introduction}
\label{sec:intro}

%% Property testing 

Property testing studies algorithms using a small number of queries to a large input that decides, with high
probability, whether the input satisfies a certain property or is far from it.  Typically, the input~$f$ is a function
mapping from a finite domain~$\domain$ to a range~$\range$.  A property~$\property$ is a subset of all such functions
$\{f: \domain \to \range\}$.  If we measure the distance between two functions by the Hamming metric, $\dist(f, g)
\defeq \Prx[x]{f(x) \neq g(x)}$, then the distance from $f$ to the property~$\property$ is $\dist(f, \property) \defeq
\min_{g \in \property} \dist(f, g)$.  Fixing a distance $\eps$, an algorithm, called a \emph{tester}, makes randomized
queries to~$f$, and outputs YES with probability at least $2/3$ for $f \in \property$, and NO with
probability at least $2/3$ if $\dist(f, \property) \geq \eps$.  A tester is said to be \emph{one-sided} if it outputs
YES with probability one for $f \in \property$.  The central question studied by property testing, as
initiated by \citet{RS96} and \citet{GGR98}, is to understand the query complexity, i.e., the minimum number of queries needed
by a tester, to test various properties.

%%  graph properties 

For example, a property is called \emph{strongly testable} if its query complexity does not depend on the size of the
doamin~$|\domain|$ and is only a function of~$\eps$.  For graph and hypergraph properties, strongly testable properties
have been exactly characterized \citep{AFNS06}.  Among strongly testable properties, it is important to
understand which ones admit testers with query complexity polynomial in $1/\eps$ and which do not.  For
example, for undirected graphs and one-sided testers, $H$-freeness for a fixed subgraph~$H$ has polynomial query
complexity if and only if $H$ is bipartite \citep{Alon02}.  Similar characterizations are known for directed graphs and
hypergraphs \citep{AS04, AS05b, RS09, AT10}.

%% Invariance and testability

\citet{KS08} suggested that symmetries, or invariance under transformations of a property, play an important role in
facilitating efficient testers.  As an easy example, a graph property, seen as functions on graph
edges, are invariant under graph isomorphisms, i.e.\@ permutations of the nodes.  \citeauthor{KS08} launched the systematic study of
\emph{algebraic} property testing, and in particular singled out \emph{linear-invariant} properties as a natural class
of properties to consider.  Restricted to the context of Boolean
functions, a property $\property \subset \{f: \Ftwon \to \{0, 1\}\}$ is said to be linear-invariant if for all $f \in
\property$ and linear transformation $L: \Ftwon \to \{0, 1\}$, the composition $f \circ L$ is still in~$\property$.  One
may further define a property~$\property$ to be \emph{linear} if it is closed under linear operations; for a property~$\property$ on
Boolean functions, this simply means $f, g \in \property$ entails $f + g \in \property$.  \citet{KS08} showed that all
properties that are linear-invariant and linear can be tested with query complexity polynomial in $1 / \eps$.
% , which unifies previous results e.g.\@ on testing linear functions \citep{BLR93, RS96} and low-degree polynomials \citep{AKK+05, JPRZ04, KR06}.  For 
When the linearity condition is relaxed, however, the picture of what is currently understood is less clear.  
\emph{Triangle-freeness} is one such property.

%% Triangle-freeness

A function $f: \Ftwon \to \{0, 1\}$ is said to be triangle-free if there are no $x, y \in \Ftwon$ such that $f(x) = f(y) = f(x + y) = 1$.  
More generally, $f$ is said to be $(M, \sigma)$-free for a fixed matrix $M \in \Ftwo^{r \times s}$ and vector $\sigma
\in \{0, 1\}^{s}$, if there exists no $x = (x_1, \ldots, x_s) \in (\Ftwon)^s$ such that $Mx = \mathbf 0$ and
$f(x_i) = \sigma_i$ for all $i \in [s]$.  \citet{Gre05} showed that $(M, \mathbf 1)$-freeness with rank-one matrix~$M$
is strongly testable (which includes triangle-freeness), and started the line of investigations resolving that any $(M,
\mathbf 1)$-freeness is strongly testable \citep{KSV13, Sha09}, and that the intersection of (possibly infinite) $(M,
\sigma)$-freeness, with rank-one~$M$, is testable \citep{BCSX11, BGS10}.  However, the upper bounds for the number of queries given in these
works, though independent of $\ncoord$, are all towers of twos whose heights are polynomial in~$1 / \eps$.  The only
exception is a result of \citet{BGRS12} showing that odd-cycle-freeness can be tested with $\tilde O(1 / \eps^2)$ queries.
It was noted by \citeauthor{BGRS12} that this property is the intersection of \emph{infinite} $(M, \mathbf 1)$-freeness.
In fact, it has been conjectured that testing any odd cycle alone takes supernomial number of queries.  Prior to this
work, the only nontrivial bound for the simplest such property, triangle-freeness, was given by \citet{BX10}, who showed
that any one-sided tester needs $\Omega(1 / \eps^{2.423})$ queries.  This is in sharp contrast with our complete understanding of
the query complexity of testing $H$-freeness in graphs, the counterpart among graph properties to $(M, \mathbf 1)$-freeness.

%% Our result & USP

\paragraph{Our Results.} In this work we improve \citet{BX10}'s lower bound and show that any one-sided tester needs
$\Omega(1 / \eps^{6.619})$ queries to test triangle-freeness in Boolean functions.  \citet{BX10}'s lower bound was built
on families of vectors having a combinatorial property called \emph{perfect-matching-free} (PMF, \autoref{def:PMF}).
Roughly speaking, a PMF family can be expanded to construct Boolean functions such that for every $x$ with $f(x) = 1$,
there exist a small number of $y$'s such that $f(y) = f(x+y) = 1$.  Such a function has a number of triangles that is
about linear with the number of $1$'s needed to be flipped to remove all triangles.  In other words, the number of
triangles is relatively small whereas the distance to triangle-freeness is relatively large, a difficult scenario for a
tester.  However, \citeauthor{BX10} were able to find only very small (and hence weak) PMF families by way of numerical calculations.
When the dimension of the family exceeds~$5$ the calculation becomes forbiddingly expensive.

In this work, we are able to construct large PMF families by using a combinatorial structure called \emph{uniquely solvable puzzles}
(USP, \autoref{def:USP}).  USP's were defined by \citet{CKSU05} in their group theoretic approach to fast matrix
multiplication.  Under their perspective,  the most important step in \citet{CW90}'s famous
$O(n^{2.376})$-time algorithm for multiplication of $n \times n$ matrices was a construction of large USP's.
\citeauthor{CW90}'s algorithm
% , indicated by the original title appropriately as ``via arithmetic progressions'',
 was for a long time the best known
algorithm for this fundamental problem, and was improved only recently \citep{Sto10, Wil12}.  As we recall in
\autoref{sec:usp-app}, \citeauthor{CW90}'s construction crucially relies on large sets of densely populated integers with no three terms in arithmetic
progressions \citep{SS42, Beh46, Elk10}.  Seen through the connection we identify here, it may not be a coincidence that
the superpolynomial lower bounds for testing nonbipartite $H$-freeness in graphs also crucially used such sets with no
arithmetic progressions \citep{Alon02}.  However, we were unable to give superpolynomial lower bounds for testing
triangle-freeness in Boolean functions.

This leads to some fascinating open problems.  For example, \citet{CKSU05} showed that, if large families of a
strengthened version of USP's, called strongly uniquely solvable puzzles (SUSP), exist, then the exponent of matrix
multiplication is 2, as has long been conjectured.  Would a large SUSP imply superpolynomial query complexity for
testing any $(M, \mathbf 1)$-freeness in Boolean functions?  On the other hand, would such a
lower bound imply the success of \citet{CKSU05}'s campaign on matrix multiplication?  We leave these questions for future investigation.

\section{Preliminaries}
\label{sec:prelim}

For an integer $n$, we let $[n]$ denote the set $\{1, 2, \ldots, n\}$.  We use $\sym(S)$ to denote the symmetric group on a set~$S$.  
We will often identify a Boolean function $f: \Ftwon \to \{0, 1\}$ with the family of subsets in $[\ncoord]$ whose
indicator function is $f$.

%% Triangle-free

We will focus on testing triangle-freeness for Boolean function triples. \footnote{This is called by \citet{BX10} the
multiple-function case.  \citet{Gre05}'s technique easily generalizes to this case, giving the same bound of tower of
twos.}
A function triple $f_1, f_2, f_3: \Ftwon \to \{0, 1\}$ is said to be triangle-free if there is no $x, y \in \Ftwon$ such that $f_1(x) = f_2(y) = f_3(x + y) = 1$.  Denote by \tfree the set of function triples that are triangle-free, and the distance of a function triple to \tfree is defined as
$$
\dist((f_1, f_2, f_3), \tfree) \defeq \min_{(g_1, g_2, g_3) \in \tfree} \dist(f_1, g_1) + \dist(f_2, g_2) + \dist(f_3, g_3).
$$

%% single-to-multiple-function reduction

As the following reduction and \autoref{thm:one-sided} shows, the multiple-function and single-function case
are essentially equivalent.\footnote{We acknowledge \citet{Xie10} for informing us of the possibility of this
reduction.}

\begin{lemma}[\citealp{Xie10}]
\label{lem:multi-single}
Given any function triple $f_1, f_2, f_3: \Ftwon \to \{0, 1\}$ which is $\eps$-far from \tfree and contains
$N$~triangles, there is a single function $f: \Ftwo^{\ncoord + 2} \to \{0, 1\}$ which is $\tfrac \eps 4$-far from triangle-freeness
and contains $N$~triangles.
\end{lemma}

\begin{proof}
Construct $f$ as follows.  For each $x \in \Ftwon$, denote by $(a, b, x)$ the $(\ncoord+2)$-dimension vector whose last
$\ncoord$ coordinates are given by~$x$.  For each $x \in \Ftwon$, let $f(0, 0, x)$ be $0$, $f(1, 0, x)$ be $f_1(x)$, $f(0, 1, x)$ be
$f_2(x)$, and $f(1, 1, x)$ be $f_3(x)$.  It is easy to see that any triangle in~$f$ has to have its three ``vertices''
given by entries from $f_1, f_2$ and~$f_3$, respectively.  The lemma follows immediately.
\end{proof}

%% lower bounds for canonical testers suffice

The \emph{canonical tester} is the naive-looking algorithm that samples $x, y \in \Ftwon$ uniformly at random and
returns YES if $f_1(x) = f_2(y) = f_3(x + y) = 1$ and NO otherwise.  A tester is said to be \emph{one-sided} if,
whenever the input satisfies the property in question, it outputs YES with probability~$1$.  By the following theorem,
it is without loss of generality to consider obfuscating the canonical tester.

\begin{theorem}[\citealp{BX10}]
\label{thm:one-sided}
Suppose there is a one-sided tester for \tfree has query complexity $q(\eps)$, then the canonical tester has query
complexity at most $O(q^2(\eps))$.  This holds for both the single-function case (when $f_1 = f_2 = f_3$) and the
multiple-function case.  
% \footnote{\citet{BX10}'s theorem is more general and applies to a family of properties called rank-$r$ matroid $\mathcal{M^*}$-freeness, of which triangle-freeness is an example.}
\end{theorem}

%% PMF

\begin{definition}[Perfect-Matching-Free (PMF) Families of Vectors]
\label{def:PMF}
Let $\ncoord$ and~$\npuzzle$ be integers such that $0 < \ncoord < \npuzzle < 2^{\ncoord}$.  
A $(\ncoord, \npuzzle)$ perfect-matching-free (PMF) family of vectors is a set of vectors $(a_i, b_i, c_i)_{i = 1}^{\npuzzle}$,
where $a_i, b_i, c_i \in \Ftwon$ and $c_i = a_i + b_i$ for all $i \in [\npuzzle]$, such that for any
permutation triple $\perm_1, \perm_2, \perm_3 \in \sym([\npuzzle])$, either $\perm_1 = \perm_2 = \perm_3$, or there
exists an $i \in [\npuzzle]$ such that $a_{\perm_1(i)} + b_{\perm_2(i)} \neq c_{\perm_3(i)}$.
\end{definition}

One can permute and then concatenate all $a_i$'s in a $(\ncoord, \npuzzle)$ PMF family and obtains $\npuzzle!$ vectors
in $\Ftwo^{\ncoord \npuzzle}$; the same can be done for $b_i$'s and $c_i$'s.  By the property of PMF, each new 
vector obtained from $a_i$'s forms one and only one triangle with two other vectors obtained from $b_i$'s and $c_i$'s,
respecitvely, and they are obtained through exactly the same permutation on $[\npuzzle]$.  This means that to remove all $\npuzzle!$ triangles in the system, one has to
remove at least the same number of vectors.  This large ratio between the distance to triangle-freeness and the
number of triangles is exactly what is needed to obfuscate a tester.  One may go further and take multiple copies of a PMF family and
repeats this experiment.  An asymptotic calculation would give the following theorem.

\begin{theorem}[\citealp{BX10}]
\label{thm:PMF-LB}
If $(\ncoord, \npuzzle)$ PMF family of vectors exists, then for small enough $\eps$ and large enough $\ncoord$,
there exists a function triple $f_1, f_2, f_3: \Ftwon \to \{0, 1\}$ that is $\eps$-far from triangle-freeness, but the
canonical tester needs $\Omega((\tfrac 1 \eps)^{\alpha})$ queries to detect a triangle, where $\alpha = (2 - \tfrac
{\log \npuzzle} {\ncoord} ) / (1 - \tfrac {\log \npuzzle} {\ncoord})$. \footnote{All logarithms in this paper are
base~$2$.}
\end{theorem}

Note that the existence of $(\ncoord, 2^{\ncoord (1 - o_{\ncoord}(1))})$ PMF family would imply a super-polynomial lower
bound for any one-sided triangle-freeness tester.

%% USP

The workhorse of our improved lower bound for testing triangle-freeness is the following combinatorial construction.  It
was implicitly developed by \citet{CW90} for their famed $O(n^{2.376})$-time matrix multiplication algrorithm, and
\citet{CKSU05} isolated it and gave it the reinterpretation we use here.

\begin{definition}[Uniquely Solvable Puzzles (USP)]
\label{def:USP}
A \emph{uniquely solvable puzzle} (USP) is a set $\usp \subset \{1, 2, 3\}^{\ncoord}$ such that, for all permutation triple
$\perm_1, \perm_2, \perm_3 \in \sym(\usp)$, either $\perm_1 = \perm_2 = \perm_3$, or there exist a $\puzzle \in \usp$
and an index $i \in [\ncoord]$ such that at least two of $(\perm_1(\puzzle))_i = 1$, $(\perm_2(\puzzle))_i = 2$ and $(\perm_3(\puzzle))_i = 3$ hold.
\end{definition}

A useful way to look at a USP is to think of it as a set of puzzles having three colors, where each color has
$\npuzzle$ pieces.  A solution to the puzzle is an arrangement of the pieces into $\npuzzle$ rows each of
size~$\ncoord$, and there cannot be a conflict.  The property in \autoref{def:PMF} requires
that there exists a unique solution to this puzzle, up to permutations on rows.

\begin{theorem}[\citealp{CW90, CKSU05}]
\label{thm:USP}
Fixing integer $\ncoord$, the largest USP is of size $\Theta((3 / 2^{2/3} - o(1))^{\ncoord})$.
\end{theorem}

The upper bound, given by an elegant construction of large USP's in \citeauthor{CW90}'s original paper was unfortunately buried in a system of
algebraic notations not easy to decipher without a proficiency with that language.  For the sake of
completeness and to promulgate this beautiful construction, we give its proof, hopefully more accessible, in
\autoref{sec:usp-app}.

\section{A Construction of PMF Families via USPs}
\label{sec:bound}

We now state the main theorem of the paper.
\begin{theorem}
\label{thm:bound}
For any $\eps > 0$ and large enough $\ncoord$,  there exists a function triple $f_1, f_2, f_3: \Ftwon \to \{0, 1\}$, such that
the triple is $\eps$-far from being triangle free, and the canonical tester needs $\Omega((\tfrac 1
\eps)^{13.239})$ queries to detect a triangle in the triple.  In addition, any one-sided tester needs $\Omega((\tfrac 1
\eps)^{6.619})$ queries.
\end{theorem}

By \autoref{thm:PMF-LB} and \autoref{thm:one-sided}, \autoref{thm:bound} would be an immediate consequence of the following lemma.

\begin{lemma}
\label{lem:large-PMF}
There exists $(\ncoord, \Theta((3 / 2^{2/3} - o(1))^{\ncoord}))$ PMF family of vectors, for all~$\ncoord$.
\end{lemma}

\begin{proof}[Proof of \autoref{lem:large-PMF}]
By \autoref{thm:USP}, it suffices to construct a $(\ncoord, |\usp|)$ PMF family for any USP $\usp \subset \{1, 2,
3\}^{\ncoord}$.  Let $\usp$ be $\{\puzzle_1, \puzzle_2, \ldots, \puzzle_{\npuzzle}\}$.  We construct $3 \npuzzle$ vectors
$a_i, b_i, c_i \in \Ftwon$ for $i = 1, 2, \ldots, \npuzzle$.  For each $i \in [\npuzzle]$, let $a_{i, j}$ be $1$ if $\puzzle_{i, j} = 1$, and $0$ otherwise; let $b_{i, j}$ be $1$ if $\puzzle_{i, j} = 2$, and $0$ otherwise; let $c_{i, j}$ be $1$ if $\puzzle_{i, j} \neq 3$, and $0$ otherwise.  It is clear now that $c_i = a_i + b_i$ for all~$i$.

We now show that $\{a_i, b_i, c_i\}_{i = 1}^{\npuzzle}$ constitutes a PMF family.  Note that a naive translation of the property of USP would not give
the desired property for PMF: for $\perm_1, \perm_2, \perm_3 \in \sym([\npuzzle])$ that are not all equal and such that
$\puzzle_{\perm_1(i),j} = 1$, $\puzzle_{\perm_2(i),j} = 2$ and $\puzzle_{\perm_3(i), j} = 3$ for some $i \in [\npuzzle],
j \in [\ncoord]$, we will have that $a_{\perm_1(i), j} = b_{\perm_2(i),j} = 1$ and $c_{\perm_3(i), j} = 0$, which does
not prevent the sum of $a_i$ and~$b_i$ from being $c_i$ in~$\Ftwon$.  Instead, we observe that for $\perm_1, \perm_2, \perm_3
\in \sym([\npuzzle])$ that are not all equal, there must be an $i \in [\npuzzle]$ and $j \in [\ncoord]$ such that
$\puzzle_{\perm_1(i), j} \neq 1$, $\puzzle_{\perm_2(i), j} \neq 2$ and $\puzzle_{\perm_3(i), j} \neq 3$: this is because of
the conservation of the total number of elements in~$\usp$.  The number of $1$'s and $2$'s and $3$'s in $\usp$ total at
$\npuzzle \ncoord$, and if, by the property according to \autoref{def:PMF}, under permutations of the puzzles there exist conflicts at some
position, then there must be some other position that is not covered by a puzzle of any color.  For such $i$ and~$j$ we
would have $a_{\perm_1(i), j} = b_{\perm_2(i), j} = 0$ and $c_{\perm_3(i), j} = 1$, which means that $a_{\perm_1(i)} +
b_{\perm_2(i)} \neq c_{\perm_3(i)}$.  This shows that we have indeed constructed a $(\ncoord, |\usp|)$ PMF family.
\end{proof}

\bibliographystyle{apalike}

\appendix

\section{Construction of Large Uniquely Solvable Puzzles}
\label{sec:usp-app}

In this appendix we present \citet{CW90}'s construction of large USP's, isolating it from the matrix multiplication context.

The construction makes use of the following theorem:

\begin{theorem}[\citealp{SS42}]
\label{thm:arithmetic}
Given $\delta > 0$, for all large enough integer~$\ringsize$, there is a set $B \subset [\ringsize]$ of size $\Omega(\ringsize^{1 - \delta})$ such that for all $b_i, b_j, b_k \in B$, $b_i + b_j \equiv 2 b_k \mod \ringsize$ if{f} $i = j = k$.
\end{theorem}

Such constructions of big sets of integers with no arithmetic progressions constitute an important class of combinatorial objects.  Improvements over \citeauthor{SS42}'s original construction with slightly larger sizes were given by \citet{Beh46} and \citet{Elk10}, but for our purpose the rougher asymptotic bound of $\Omega(\ringsize^{1 - \delta})$ suffices.

Now we are ready to describe the construction.  We fix a large enough integer $\puzsize$ and $\ringsize = 2 {2\puzsize
\choose \puzsize} + 1$.  Fix $B \subset [\ringsize]$ as given by \autoref{thm:arithmetic}.  Sample $3 \puzsize$ integers
$0 \leq w_j < \ringsize$ independently at random for each $j = 0, 1, \cdots, 3 \puzsize$.  We will call these $w_j$'s
\emph{weights}.  Now
consider the set $\puzset$ of all subsets $\puz \subset [3\puzsize]$ of size $\puzsize$.  Let $\delta_{\puz}$ be the
indicator function of subset~$\puz$, i.e., for each $j \in [3\puzsize]$, $\delta_{\puz}(j) = 1$ for $j \in \puz$, and $0$ otherwise.  The weights we sampled define three mappings from $\puzset$ to $\mathbb Z_{\ringsize}$:
\begin{align}
\beta_{x}(\puz) & \equiv \sum_{j = 1}^{3 \puzsize} \delta_{\puz}(j) w_j \mod \ringsize; \\
\beta_{y}(\puz) & \equiv w_0 + \sum_{j = 1}^{3 \puzsize} \delta_{\puz}(j) w_j \mod \ringsize; \\
\beta_{z}(\puz) & \equiv \left( w_0 + \sum_{j = 1}^{3 \puzsize} (1 - \delta_{\puz}(j)) w_j \right) / 2 \mod \ringsize.
\end{align}
Note that the operation of division by~$2$ is well defined for $\beta_z$, as $\ringsize$ is odd.

With these mappings, we will consider each element $b_i \in B$.  First, with each $b_i \in B$ we associate all triples $(I, J,
K)$, where $I, J, K \in \puzset$ are pairwise disjoint, and $\beta_x(I) = \beta_y(J) = \beta_z(K) = b_i$.  (A triple $(I, J, K)$
is discarded if the members are not pairwise disjoint, or if they are not mapped to be same~$b_i$.) Second, among all
triples associated with the same~$b_i$, we arbitrarily remove all but one triple.  To construct our USP $\usp \subset
\{1, 2, 3\}^{3\puzsize}$, there will be a puzzle $\puzzle_i$ for each $b_i$ associated with a nonempty triple $(I_i,
J_i, K_i)$, and for each $j \in [3\puzsize]$, $\puzzle_i(j) = 1$ for $j \in I_i$, $\puzzle_i(j) = 2$ for $j \in J_i$,
and $\puzzle_i(j) = 3$ for $j \in K_i$.

We first check that we indeed obtain a USP family, before going on to prove its expected size.

\begin{claim}
\label{claim:no-intersection}
For any $i_1, i_2, i_3 \in [|B|]$, $I_{i_1}, J_{i_2}$ and $K_{i_3}$ are pairwise disjoint if{f} $i_1 = i_2 = i_3$.
\end{claim}

Note that \autoref{claim:no-intersection} suffices for the property of USP (\autoref{def:USP}).

\begin{proof}
Suppose $I_{i_1}, J_{i_2}$ and $K_{i_3}$ are pairwise disjoint, we have that
\begin{align}
b_{i_1}& \equiv \beta_{x}(i_1)  \equiv \sum_{j = 1}^{3 \puzsize} \delta_{I_{i_1}}(j) w_j \mod \ringsize; \\
b_{i_2}& \equiv \beta_{y}(i_2)  \equiv w_0 + \sum_{j = 1}^{3 \puzsize} \delta_{J_{i_2}}(j) w_j \mod \ringsize; \\
b_{i_3}& \equiv \beta_{z}(i_3)  \equiv \left( w_0 + \sum_{j = 1}^{3 \puzsize} (1 - \delta_{K_{i_3}}(j)) w_j \right) / 2
\equiv \left( w_0 + \sum_{j = 1}^{3 \puzsize} \delta_{I_{i_1} \cup J_{i_2}} w_j \right) / 2 \mod \ringsize.
\end{align}
Straightforwardly, we will have $b_{i_1} + b_{i_2} - 2b_{i_3} \equiv 0 \mod \ringsize$.  However, since $b_{i_1},
b_{i_2}$ and $b_{i_3}$ are in~$B$, by the property of~$B$, it can only be that $i_1 = i_2 = i_3$.
\end{proof}

We now show that the we indeed have a large USP.  This amounts to showing that we have many triples left at the end of
the second step of the construction.  We first consider the number of triples associated with elements in~$B$ in the first
step.

\begin{claim}
\label{claim:size-initial}
Fixing $b_i \in B$, the expected number of triples $(I, J, K)$ associated with~$b_i$ in the first step is ${3\puzsize \choose \puzsize, \puzsize, \puzsize}
\ringsize^{-2}$.
\end{claim}

\begin{proof}
First, by the same calculation as in \autoref{claim:no-intersection}, we know that if two disjoint $I, J \in \puzset$
are mapped to the same $b_i \in B$ by $\beta_x$ and $\beta_y$, respectively, then their complement, $K = [3\puzsize] - (I
\cup J)$, must be mapped to be same~$b_i$ by $\beta_z$.  Now there are ${3\puzsize \choose \puzsize, \puzsize,
\puzsize}$ disjoint triples, the probability that each of the two components is mapped to a fixed $b_i$ is
$\ringsize^{-1}$, respectively.  Moreover, the two events are independent.  The claim follows.  
\end{proof}

\begin{claim}
\label{claim:size-removal}
Fixing $b_i \in B$, the expected number of triples $(I, J, K)$ associated with~$b_i$ that we remove in the second step
is at most $\frac 3 2 {3\puzsize \choose \puzsize, \puzsize, \puzsize} ({2\puzsize \choose \puzsize} - 1) \ringsize^{-3}$.
\end{claim}

\begin{proof}
Fixing $b_i \in B$, the expected number of triples $(I, J, K)$ and $(I', J', K)$ ($I \neq I')$ associated with $b_i$ is
$\frac 1 2 {3\puzsize \choose \puzsize, \puzsize, \puzsize} ({2\puzsize \choose \puzsize} - 1) \ringsize^{-3}$.  The
term ${2\puzsize \choose \puzsize} - 1$ counts the number of $I'$'s disjoint with~$K$ and unequal to~$I$.  The factor
$\tfrac 1 2$ disregards the symmetric case $(I, J, K), (I', J', K)$ and $(I', J', K), (I, J, K)$.  The additional
$\ringsize^{-1}$ here (as compared to the count in \autoref{claim:size-initial}) indicates the probability of the event $\beta_y(I') = b_i$.  Note that this event is independent
from the events $\beta_x(I) = b_i$ and $\beta_y(J) = b_i$, even if $J'$ can be equal to $I$, because of the presence of
the weight~$w_0$ in the definition of~$\beta_y$.  Repeat the argument for the cases when two triples coincide on the
first or second subset, and the claim follows.  (The event that two triples associated with the same $b_i$ disagree on
each subset they contain is neglected here, since its probability is significantly smaller than that of the case analyzed
here.  For large $\puzsize$ and~$\ringsize$ this is easily accommodated.)
\end{proof}

Therefore, by our choice of~$\ringsize$, the expected number of triples associated with each $b_i$ remaining after the second step is at least 
$$
{3\puzsize \choose \puzsize, \puzsize, \puzsize} \ringsize^{-2} - \frac 3 2 {3\puzsize \choose \puzsize, \puzsize,
\puzsize} \left({2\puzsize \choose \puzsize} - 1 \right) \ringsize^{-3} \geq \frac 1 4 {3\puzsize \choose \puzsize, \puzsize,
\puzsize} \ringsize^{-2}.
$$

With a standard probabilistic argument, we conclude that there exists a choice of $w_j$'s such that the size of USP we
obtain is at least
$$
\frac 1 4 {3\puzsize \choose \puzsize, \puzsize, \puzsize} \ringsize^{-2} |B| = 
\frac 1 4 {3\puzsize \choose \puzsize, \puzsize, \puzsize} \ringsize^{-2} \ringsize^{1 - \delta}.
$$

Substituting our choice of $\ringsize$ and applying the Stirling's formula, we get the desired asymptotic bound of $(3 / 2^{2/3}
- o(1))^{3 \puzsize}$ for the size of USP.
\end{document}